\documentclass[11pt, a4paper]{article}
\usepackage{anyfontsize} 
\usepackage{amssymb}
\usepackage{amsmath}
\usepackage{amsthm}
\usepackage{amsfonts}
\usepackage{graphicx}
\usepackage{epstopdf}
\usepackage{pdflscape}
\usepackage{float}
\usepackage{latexsym}
\usepackage{color}
\usepackage{amsmath,amssymb}
\usepackage[T1]{fontenc}
\usepackage[utf8]{inputenc}
\usepackage{authblk}
\usepackage{natbib}
\usepackage{fancyhdr}
\usepackage{hyperref}
\hypersetup{colorlinks,citecolor=blue}
\usepackage{multirow}
\usepackage{multicol}
\usepackage{titlesec}
\usepackage{caption} 
\titlelabel{\thetitle.\quad}
\usepackage[margin=0.8in]{geometry}
\date{}

\newtheorem{theorem}{Theorem}[section]

\newtheorem{prop}{Proposition}[section]
\newtheorem{rem}{Remark}[section]

\title{Analyzing health care data using count models: A novel approach to Length of Stay analysis}

\author{Peer Bilal Ahmad$^{\it{1,*}}$}
\author{Na Elah$^{\it{2}}$}
\affil{$^{\it{1,*}}$Islamic University of Science and Technology, Kashmir. E-mail: \href{mailto:bilalahmadpz@gmail.com}{bilalahmadpz@gmail.com}}
\affil{$^{\it{2}}$Islamic University of Science and Technology, Kashmir. E-mail: \href{mailto:naaelashah@gmail.com}{naaelashah@gmail.com}}

\date{}

\begin{document}

\maketitle

\begin{abstract}

\noindent Count data modeling has been extensively applied in medical sciences to analyze various healthcare datasets. Numerous probability models have been developed to address diverse aspects of healthcare data. In this study, we propose a novel count data model for analyzing healthcare datasets. Key structural properties of the model are established, and an associated regression framework is introduced to examine the effects of various covariates. Additionally, a three-inflated distribution, based on the proposed model, is presented to analyze length of stay of patients in hospitals.

\end{abstract}

\section*{Statements and Declarations}
\textbf{Conflict of Interest:} On behalf of all authors, the corresponding author states that there is no conflict of interest.

\section*{Acknowledgement}
This work is supported by Research project (JKST\&IC/SRE/906-10) provided by JKST\&IC, UT of J\&K, India.\\
\textbf{Funding:} None

\section*{Classification code}
\textbf{MSC:} 60E05, 62F03, 62J05, 62P10

\newpage
\section{Introduction}\label{sec1}
In many fields, including medical science, count data—which consists of non-negative integer values—is used to analyse events like the number of hospital admissions, prescription dosages, or days spent in medical attention. Traditional statistical methods are not well-suited to handle the difficulties posed by the discrete and  skewed nature of count data. As a result, specific techniques have been created to efficiently model and analyse count data, offering important insights into underlying trends and connections.\\
\indent Regression models like the Poisson and Negative Binomial are essential tools in the data science era, and count data modelling has changed dramatically \citep{Cameron13}, \citep{Hilbe11}. These models take into consideration features like equi-dispersion and over-dispersion and make it easier to investigate relationships between count outcomes and explanatory variables. Nevertheless, a lot of real-world datasets show extra complexity, such as a high number of zeros or variation in the methods used to generate the data. These difficulties have prompted the creation of sophisticated models that go beyond the capability of conventional regression techniques.\\
\indent Inflated count data models, including Zero-Inflated Poisson (ZIP) and Zero-Inflated Negative Binomial (ZINB) are two examples that have become popular for handling particular data structures where standard and classical models are inadequate. These models handle three crucial situations that are frequently seen in count data: Over-dispersion (variance exceeds the mean), Zero Inflation (excess of zero observations) and Hurdle effects (where zero and non zero outcomes are governed by separate mechanisms).\\
\indent One important application of these models is analysis of length of stay (LOS) in hospitals, which is a crucial statistic in healthcare research.  Understanding patient outcomes, resource usage, and hospital management tactics all depend heavily on LOS data. However, LOS data often presents challenges, such as a tendency towards brief stays brought on by standard treatments or longer stays brought on by difficulties. So, these features make it perfect fit for advanced count data modelling.

\citep{Du12} investigated the use of several count data models to analyse the issues in health information technology (HIT) evaluations. They demonstrated that the average number of laboratory orders per patient per diagnosis-related category increased by 50\% within two years of the implementation of the computerised physician order entry system. \citep{Ahmad23} introduced a new compound model and showed its effectiveness in health sciences by applying this model to patient data set receiving anti-convulsant drugs. \citep{Skinder23} introduces a new zero inflated count model and showed its applicability in medical sciences. \citep{AltunE21} introduced two parameter discrete poisson-generalised Lindley distribution using mixed-poisson technique. \citep{AltunE21} also developed the associated regression model and analysed LOS of patient dataset. Using the model, \citep{AltunE21} demonstrated that if a patient has a CABG surgery, admitted on urgent basis or is older than 75, their length of hospital stay increases. Using the same approach, \citep{AltunE21} also examined the data set on physician visits. He deduced from this that if a patient is female, unmarried, in bad health, and employed, the frequency of doctor visits increases. \citep{AltunE21} also demonstrated that if a patient is too old, fewer doctor visits occur. \citep{Austin03} compared different count regression models for analysing the cost per patient undergoing CABG surgery. \citep{Wani23} introduced a new count regression model and analysed a health care data set related to the number of infected blood cells in a patient. They concluded that smoking and being over-weight increases the number of infected blood cells. \citep{Austin02} investigated the effectiveness of several analytical techniques for predicting hospital LOS and evaluating the relationship between patient characteristics and post-operative LOS. \citep{Morshedy22} introduced discrete generalized Lindley, a count data model and analyzed the number of daily coronavirus cases and deaths. \citep{Shahsavari24} introduced Robust Zero Inflated Poisson (RZIP) regression model and analysed the LOS of 254 patients who were admitted to ICU. They showed that the factors including age, underlying medical conditions and having insurance, significantly impact LOS. 

This research also explores the theoretical underpinnings and practical applications of count data regression models, with a particular emphasis on count data in health sciences.  By applying these models to health care data, this study aims to provide novel insights into key factors influencing healthcare outcomes, improve predictive capabilities, and contribute to more effective healthcare planning and decision-making.

The rest of the paper is organized as follows: Section \ref{sec:cd} discusses the two-parameter Copoun distribution. Section \ref{sec:pcd} provides a detailed introduction to the Poisson Copoun distribution, including the derivation of its structural properties. Section \ref{sec:est} outlines the estimation techniques for the proposed distribution including maximum likelihood estimation and the method of moments. Section \ref{sec:lm} introduces the linear regression model based on the Poisson Copoun distribution including the residual analysis. Section \ref{sec:ThI} presents an inflated version of the proposed model, while Section \ref{sec:dp} explores the structural properties of the inflated model in detail. Section  \ref{sec:mle} focuses on the maximum likelihood estimation for the inflated version of the proposed model. Section \ref{sec:app} highlights practical applications of the proposed model, its associated regression model, and the three-inflated version. Finally, Section \ref{sec:con} offers some concluding remarks on the proposed study.

\section{Copoun Distribution} \label{sec:cd}
\citep{OR23} introduced two parameter Copoun distribution (CD) with its probability density function (pdf) and cummulative distribution function (cdf) given as:
\begin{equation}
g(x;\eta,\phi)=\frac{\eta^2}{(\phi+\eta)}\left[1+\frac{\phi\eta^2 x^3}{6}\right]e^{-\eta x}  \;; x>0\;; \;\eta,\phi>0 \label{pdf:cd}
\end{equation}
\begin{equation}
F(x)=1-\left[1+\frac{\phi\eta^3 x^3+3\phi\eta^2 x^2+6\phi\eta x}{6(\phi+\eta)}\right]e^{-\eta x}
\end{equation}
With the mixing proportion $\pi=\frac{\eta}{(\phi+\eta)}$, the CD is a combination of Exponential($\eta$) and Gamma(4,$\eta$). Accordingly, the properties of the Exponential and Gamma distributions can be used to deduce the major properties of the CD (see \citep{Altun21}).

\section{Poisson Copoun Distribution} \label{sec:pcd}
Using the compoundig technique, a novel two parameter discrete distribution is introduced in this section.

\begin{prop}
Let the random variable X follows Poisson Copoun distribution (PCD) which holds the following stochastic representation
\begin{align*}
X|\lambda \sim Poisson(\lambda) \\
\lambda|\eta,\phi \sim CD(\eta,\phi)
\end{align*}
where $\eta,\phi$ and $\lambda>0$, then the probability mass function(pmf) of the unconditional distribution of X is
\begin{align}
p(x;\eta,\phi)=\frac{\eta^2}{(\phi+\eta)(1+\eta)^{x+4}}\left[(1+\eta)^3+\frac{\phi\eta^2}{6}(x+1)(x+2)(x+3)\right] \; , \;x=0,1,2,3,... \label{pmf:PCD}
\end{align}
\end{prop}

\begin{proof}
The pmf of X is derived as follows

\begin{align}
p(x;\eta,\phi)&=\int\limits_0^\infty {\frac{e^{-\lambda}{ {\lambda}^x}}{x!}\times\frac{\eta^2}{(\phi+\eta)}\left[1+\frac{\phi\eta^2 \lambda^3}{6}\right]e^{-\eta \lambda} d\lambda} \nonumber\\  
=&\frac{\eta^2}{(\phi+\eta)}\left[\int\limits_0^\infty\frac{e^{-(1+\eta)\lambda}\lambda^x}{x!}d\lambda +\frac{\phi\eta^2}{6}\int\limits_0^\infty\frac{e^{-\lambda(1+\eta)}\lambda^{x+3}}{x!}d\lambda\right]\nonumber\\
=&\frac{\eta^2}{(\phi+\eta)x!}\left[\frac{\Gamma(x+1)}{(1+\eta)^{x+1}}+\frac{\phi\eta^2}{6}\frac{\Gamma(x+4)}{(1+\eta)^{x+4}}\right]\nonumber\\
=&\frac{\eta^2}{(\phi+\eta)(1+\eta)^{x+4}}\left[(1+\eta)^3+\frac{\phi\eta^2}{6}(x+1)(x+2)(x+3)\right]
\end{align}
\end{proof}

\begin{rem}
\noindent When $\phi$=0, PCD reduces to geometric distribution with $p=\eta/(1+\eta)$.
\end{rem}

The cdf of PCD is given as:
\begin{equation}
F(x)=1-\frac{(x^3+9x^2+26x+24)\phi\eta^5+(3x^2+33x+24)\phi\eta^4+6(x+4)\phi\eta^3+6\phi\eta^2+6\eta^3(1+\eta)^3}{6\eta^2(\phi+\eta)(1+\eta)^{x+4}}
\end{equation}

\begin{figure}[H]
\centering
\includegraphics[width=15cm,height=10cm]{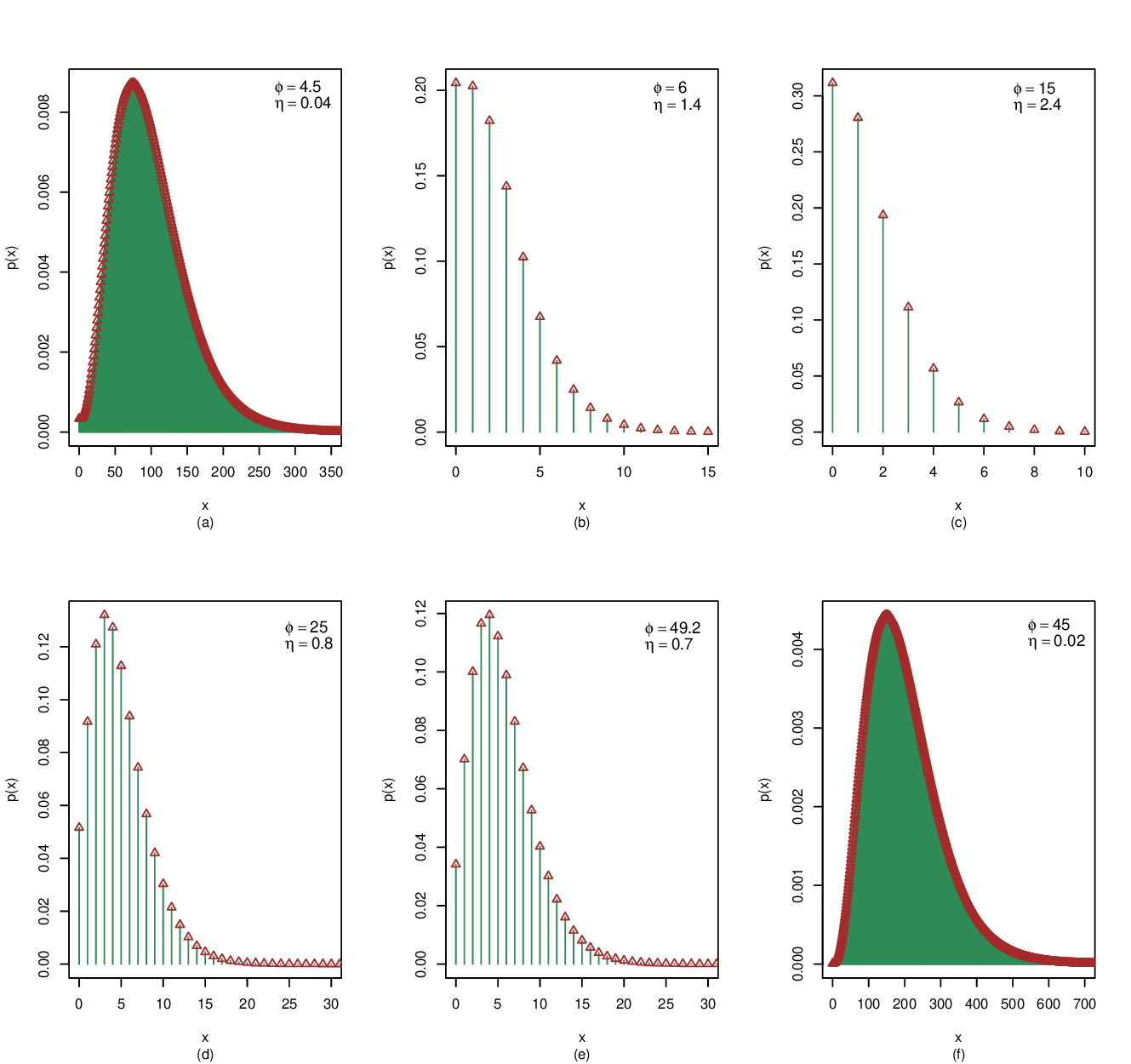}
\caption{The pmf plots of PCD}
\end{figure}

\begin{prop}
	The factorial moments of X are given by
\begin{equation}
\mu^{[r]}=\frac{r!}{\eta^r(\phi+\eta)}\left[\eta+\frac{\phi}{6}(r+1)(r+2)(r+3)\right] \label{eq:fact}
\end{equation}
\end{prop}
\begin{proof}
The factorial moments of X are derived as follows
\begin{eqnarray}
\mu^{[r]}&=&\frac{\eta^2}{(\phi+\eta)}\int_{0}^{\infty}\lambda^r\left[1+\frac{\phi\eta^2\lambda^3}{6}\right]e^{-\eta\lambda} d\lambda \nonumber\\
&=&\frac{\eta^2}{(\phi+\eta)}\left[\int_{0}^{\infty}\lambda^r e^{-\eta\lambda} d\lambda+\frac{\phi\eta^2}{6}\int_{0}^{\infty}\lambda^{r+3}e^{-\eta\lambda} d\lambda\right] \nonumber\\
&=&\frac{\eta^2}{(\phi+\eta)}\left[\frac{\Gamma{(r+1)}}{\eta^{r+1}}+\frac{\phi\eta^2}{6}\frac{\Gamma{(r+4)}}{\eta^{r+4}}\right]\nonumber\\
&=&\frac{r!}{\eta^r(\phi+\eta)}\left[\eta+\frac{\phi}{6}(r+1)(r+2)(r+3)\right]\nonumber
\end{eqnarray}\end{proof}

Using (\ref{eq:fact}), we have
\begin{eqnarray}
E(X)&=&\frac{\eta+4\phi}{\eta(\phi+\eta)}\nonumber\\
E(X^2)&=&\frac{\eta^2+2(2\phi+1)\eta+20\phi}{\eta^2(\phi+\eta)}\nonumber\\
E(X^3)&=&\frac{\eta^3+2(2\phi+3)\eta^2+6(10\phi+1)\eta+120\phi}{\eta^3(\phi+\eta)}\nonumber\\
E(X^4)&=&\frac{\eta^4+2(2\phi+7)\eta^3+4(35\phi+9)\eta^2+24(30\phi+1)\eta+840\phi}{\eta^4(\phi+\eta)}\nonumber
\end{eqnarray}
The variance of PCD is given by
\begin{equation}
Var(X)=\frac{\eta^3+(5\phi+1)\eta^2+2\phi(2\phi+7)\eta+4\phi^2}{\eta^2(\phi+\eta)^2}
\end{equation}
The dispersion index ($DI$) of PCD is given as
\begin{equation}
DI=1+\frac{\eta^2+14\phi\eta+4\phi^2}{\eta(\phi+\eta)(\eta+4\phi)} > 1 \label{eq:dipcd}
\end{equation}
%From (\ref{eq:dipcd}), it is obvious that PCD is over-dispersed as can be seen form Figure (\ref{fig:iodpcd}).

%\begin{figure}[H]
%\centering
%\includegraphics[width=12cm,height=8.5cm]{iodpcd.png}
%\caption{The iod plot of PCD} \label{fig:iodpcd}
%\end{figure}

\begin{prop}
	The probability generating function (pgf) of X is
\begin{align}
P(s)= \frac{\eta^2}{\phi+\eta}\left[\frac{(\eta-s+1)^3+\phi\eta^2}{(\eta-s+1)^4}\right]    \label{eq:pgfPCD}
\end{align}	
\end{prop}

\begin{proof}
The pgf of PCD is derived as follows
\begin{align}
P(s)&=E(s^x)\nonumber \\
=&\int_{0}^{\infty}e^{\lambda(s-1)}\frac{\eta^2}{\phi+\eta}\left[1+\frac{\phi\eta^2\lambda^3}{6}\right]e^{-\eta\lambda}d\lambda \nonumber \\
=&\frac{\eta^2}{\phi+\eta}\Bigg[\int_{0}^{\infty}e^{\lambda(s-1-\eta)}d\lambda+\frac{\phi\eta^2}{6}\int_{0}^{\infty}\lambda^3e^{\lambda(s-1-\eta)}d\lambda\Bigg] \nonumber\\
=&\frac{\eta^2}{\phi+\eta}\Bigg[\frac{1}{(\eta-s+1)}+\frac{\phi\eta^2}{(\eta-s+1)^4}\Bigg] \nonumber\\
=&\frac{\eta^2}{\phi+\eta}\Bigg[\frac{(\eta-s+1)^3+\phi\eta^2}{(\eta-s+1)^4}\Bigg] \nonumber
\end{align}
Hence proved.
\end{proof}

\begin{rem}
Putting $S=e^t$ in (\ref{eq:pgfPCD}), the Moment Generating Function $M_X(t)$ of X $\sim$ PCD($\eta,\phi$) is as follows
\begin{equation}
M_X(t)= \frac{\eta^2}{\phi+\eta}\left[\frac{(\eta-e^t+1)^3+\phi\eta^2}{(\eta-e^t+1)^4}\right]    \label{eq:mgf}
\end{equation}
\end{rem}
 \begin{rem}
 Putting $S=e^{it}$ in (\ref{eq:pgfPCD}), the Characteristic Function $\varphi_X(t)$ of X $\sim$ PCD($\eta,\phi$) is as follows
\begin{equation}
\varphi_X(t)= \frac{\eta^2}{\phi+\eta}\left[\frac{(\eta-e^{it}+1)^3+\phi\eta^2}{(\eta-e^{it}+1)^4}\right] 
\end{equation}
\end{rem}

\section{Estimation} \label{sec:est}
\subsection{Method of Moments}
This approach uses the method of moments(MoM) to estimate the parameters of PCD model. The concept of this method is to use empirical moments to solve theoretical moments. Thus, we employ $m_1$ for the first sample moment and $m_2$ for the second sample moment. With this idea, we have:

\begin{eqnarray}
\mu_1'&=&\frac{\eta+4\phi}{\eta(\phi+\eta)} =m_1 \label{eq:mom1}\\ 
\mu_2'&=&\frac{\eta^2+2(2\phi+1)\eta+20\phi}{\eta^2(\phi+\eta)}=m_2 \label{eq:mom2}
\end{eqnarray}
The joint solution of (\ref{eq:mom1}), and (\ref{eq:mom2}) yields the following results.
\begin{equation}
\hat\phi_{MM}=\frac{m_1\eta^2-\eta}{4-m_1\eta}
\end{equation}

\begin{equation}
\hat\eta_{MM}=\frac{3m_1+\sqrt{9m_1^2+4(m_1-m_2)}}{m_2-m_1}
\end{equation}

\begin{prop}
	For fixed $\phi$, the MoM estimator $\hat\eta$ of $\eta$ is positively biased.
\end{prop}
\begin{proof}
Assume $\hat\eta_{MM}=h(\bar x)$ and ${h}(t)=\Big(\sqrt{\phi^2t^2+14\phi\eta+1}-\phi t+1\Big)/(2t)$ for t > 0. The second derivative of h(t) is	
\begin{equation}
{h}''(t)=\frac{7\phi^3t^3+75\phi^2t^2+21\phi t^2+1+[\phi^2t^2+14\phi t+1]^{3/2}}{t^3[\phi^2t^2+14\phi t+1]^{3/2}}>0.
\end{equation}
$\Rightarrow $ $h(t)$ is strictly convex. Therefore, by means of Jensen's inequality,
\begin{equation}
E(h(\bar X))>h(E(\bar X))
\end{equation}
we get $E(\hat\eta_{MM})> \eta $ where $h(E(\bar X))=h(\mu)=h((\eta+4\phi)/(\eta(\phi+\eta)))=\eta$.
\end{proof}

\begin{prop}
	For fixed $\phi$, the MoM estimator is consistent and asymptotically normal.
	\begin{align}
	\sqrt{n}(\hat\eta-\eta)\xrightarrow{\text{d}} N(0,\nu^2(\eta)),
	\end{align}
	where
	\begin{align}
	\nu^2(\eta)=\frac{\begin{aligned}&\eta^2(\phi+\eta)^2[\eta^3+(5\phi+1)\eta^2\\
	&\hspace{2cm}+2\phi(2\phi+7)\eta+4\phi^2]\end{aligned}}{(\eta^2+16\phi^2+8\phi\eta)^2}
	\end{align}
\end{prop}

\begin{proof}
Using delta method,
\begin{align}
\sqrt{n}(h(\bar x)-h(\mu))\xrightarrow{\text{d}} N\Big(0,[\textsl{h}'(\mu)]^2\sigma^2\Big),
\end{align}
where $h(\bar x)=\hat\eta$ , $h(\mu)=\eta$ and
\begin{align}
h'(\mu)=-\frac{\eta^2(\eta+\phi)^2}{\eta^2+4\phi^2+8\phi\eta}
\end{align}							
\end{proof}

\subsection{Maximum Likelihood Estimation}
The log-likelihood function of the PCD is
\begin{equation}
L(\eta,\phi)=2n\ln\eta-n\ln(\phi+\eta)-\sum_{i=1}^{n}(x_i+4)\ln(1+\eta)+\sum_{i=1}^{n}\ln\left[(1+\eta)^3+\frac{\phi\eta^2}{6}(x_i+1)(x_i+2)(x_i+3)\right]  
\end{equation}
\begin{equation}
\Rightarrow \frac{\partial L}{\partial\eta}=\frac{2n}{\eta}-\frac{n}{\phi+\eta}-\frac{\sum_{i=1}^{n}(x_i+4)}{1+\eta}+\sum_{i=1}^{n}\left[\frac{3(1+\eta)^2+\frac{\phi\eta}{3}(x_i+1)(x_i+2)(x_i+3)}{(1+\eta)^3+\frac{\phi\eta^2}{6}(x_i+1)(x_i+2)(x_i+3)}\right] 
\end{equation}\\

\begin{equation}
\&\; \frac{\partial L}{\partial\phi}=\frac{-n}{\phi+\eta}+\sum_{i=1}^{n}\frac{\frac{\eta^2}{6}(x_i+1)(x_i+2)(x_i+3)}{(1+\eta)^3+\frac{\phi\eta^2}{6}(x_i+1)(x_i+2)(x_i+3)}
\end{equation}
The log-likelihood function is directly maximised to obtain the MLEs of $\eta$ and $\phi$. The asymptotic confidence intervals of $\hat\eta$ and $\hat\phi$ are
\begin{multicols}{2}
  \begin{equation}
    \hat\eta\pm z_{p/2}\sqrt{\widehat{Var(\hat\eta)}} \nonumber
  \end{equation}\break 
  \begin{equation} 
    \hat\phi\pm z_{p/2}\sqrt{\widehat{Var(\hat\phi)}} \nonumber
  \end{equation}
\end{multicols}
where $z_{p/2}$ represents the upper ${p/2}$ quantile of the standard normal distribution.

\section{PCD Linear Model} \label{sec:lm}
The PCD is over-dispersed since its dispersion is more than one (see \ref{eq:dipcd}). Therefore, PCD can be used in place of the NB distribution for over-dispersed data sets. This section introduces a new count regression model based on the PCD

Let $\eta=\left(1-\phi\mu+\sqrt{(\phi\mu-1)^2+16\phi\mu}\right)/({2\mu})$, then the pmf of $Y \sim PCD(\mu,\phi)$ is given by
\begin{align}
p(y|\phi,\mu)&=\frac{\Bigg(\frac{1-\phi\mu+\sqrt{(\phi\mu-1)^2+16\phi\mu}}{2\mu}\Bigg)^2}{\left(\phi+\frac{1-\phi\mu+\sqrt{(\phi\mu-1)^2+16\phi\mu}}{2\mu}\right)\left(1+\frac{1-\phi\mu+\sqrt{(\phi\mu-1)^2+16\phi\mu}}{2\mu}\right)^{y+4}}
\nonumber\\   &\times
\left[\left(1+\frac{1-\phi\mu+\sqrt{(\phi\mu-1)^2+16\phi\mu}}{2\mu}\right)^3 \nonumber \right.\\   &\left.
+\frac{\phi}{6}\left(\frac{1-\phi\mu+\sqrt{(\phi\mu-1)^2+16\phi\mu}}{2\mu}\right)^2(y+1)(y+2)(y+3)\right] \label{eq:repar}
\end{align}
where $\phi>0, \mu>0$ and $E(Y_i|\phi,\mu_i)=\mu_i$. The covariates are linked to the mean of the response variable $y_i$, using the below log-link function
\begin{equation}
\mu_i=\exp(\textbf{x}_i^T\beta) \quad, i=1,2,...,n  \label{eq:cov}
\end{equation}
where $\textbf{x}_i^T=(x_{1i}, x_{2i},..., x_{ki})$ is the vector of covariates and $\beta=(\beta_0, \beta_1,...,\beta_k)^T$ is the unknown vector of regression coefficients. For details on link function see (\citep{McCullagh19}). \\
Substituting (\ref{eq:cov}) in (\ref{eq:repar}), the pmf of $Y_i|x_i^T \sim PCD(\phi,\mu_i)$ is defined as a linear model.\begin{align}
p(y_i|x_i^T,\phi)&=\frac{\Bigg(\frac{1-\phi\exp(x_i^T\beta)+\sqrt{(\phi\exp(x_i^T\beta)-1)^2+16\phi\exp(x_i^T\beta)}}{2\exp(x_i^T\beta)}\Bigg)^2}{\left(\phi+\frac{1-\phi\exp(x_i^T\beta)+\sqrt{(\phi\exp(x_i^T\beta)-1)^2+16\phi\exp(x_i^T\beta)}}{2\exp(x_i^T\beta)}\right)} \nonumber \\
& \times \frac{1}{\left(1+\frac{1-\phi\exp(x_i^T\beta)+\sqrt{(\phi\exp(x_i^T\beta)-1)^2+16\phi\exp(x_i^T\beta)}}{2\exp(x_i^T\beta)}\right)^{y_i+4}}
\nonumber\\   &\times
\left[\left(1+\frac{1-\phi\exp(x_i^T\beta)+\sqrt{(\phi\exp(x_i^T\beta)-1)^2+16\phi\exp(x_i^T\beta)}}{2\exp(x_i^T\beta)}\right)^3 \right. \nonumber \\ & \left.
+\frac{\phi}{6}\left(\frac{1-\phi\exp(x_i^T\beta)+\sqrt{(\phi\exp(x_i^T\beta)-1)^2+16\phi\exp(x_i^T\beta)}}{2\exp(x_i^T\beta)}\right)^2(y_i+1)(y_i+2)(y_i+3)\right] \label{eq:reg}
\end{align}

The unknown parameters ($\phi, \beta$) of the model can be estimated using the MLE method. 

\subsection{Model estimation}
The MLE approach is used to estimate the regression parameters $\beta$ and the distribution parameter $\phi$. The log-likelihood function of PCD linear model is:
\begin{equation}
\ell(\Theta)= 
          \sum_{i=1}^{n}\left\{\!\begin{aligned}
          &2\ln\left(1-\phi\exp({x_{i}^{T}\beta})+\sqrt{(\phi\exp({x_{i}^{T}\beta})-1)^2+16\phi\exp({x_{i}^{T}\beta})}\right)+y_i\ln2+y_i{x_{i}^{T}\beta}-\ln24\\[1ex]
         & -(y_i+4)\ln\left(2\exp({x_{i}^{T}\beta})+1-\phi\exp({x_{i}^{T}\beta})+\sqrt{(\phi\exp({x_{i}^{T}\beta})-1)^2+16\phi\exp({x_{i}^{T}\beta})}\right)\\[1ex]
         &-\ln\left(1+\phi\exp({x_{i}^{T}\beta})+\sqrt{(\phi\exp({x_{i}^{T}\beta})-1)^2+16\phi\exp({x_{i}^{T}\beta})}\right)\\[1ex]
         &+\ln\Bigg[3\left(2\exp({x_{i}^{T}\beta})+1-\phi\exp({x_{i}^{T}\beta})+\sqrt{(\phi\exp({x_{i}^{T}\beta})-1)^2+16\phi\exp({x_{i}^{T}\beta})}\right)^3 \Bigg.\\[1ex]
        & \Bigg. +\phi\exp({x_{i}^{T}\beta})\left(1-\phi\exp({x_{i}^{T}\beta})+\sqrt{(\phi\exp({x_{i}^{T}\beta})-1)^2+16\phi\exp({x_{i}^{T}\beta})}\right)^2  \Bigg.\\[1ex]
       & \Bigg.\hspace{8cm} (y_i+1)(y_i+2)(y_i+3)\Bigg]
 \end{aligned}\right\} \label{eq:logregpc}
\end{equation}   
where $\Theta$=($\phi$, $\beta$) is the unknown parameter vector. The $nlm$ function of R software is used to maximise (\ref{eq:logregpc}) in order to estimate the unknown parameters, $\phi$ and $\beta$. The standard errors of the estimated parameters are obtained by taking the inverse of the observed information matrix.

\subsection{Residual Analysis}
\citep{Dunn96} proposed the randomised quantile residuals (rqrs) to assess how well the model fit the data set.\\Let the cdf of the mean-parametrized PCD be $F(y;\alpha,\mu)$, then the rqrs are given by 
\begin{equation} 
r_{q,i}=\Phi^{-1}(u_i), 
\end{equation}
where $u_i=F(y_i;\hat\mu_i)$ is uniformly distributed between $a_i=\lim_{y \uparrow y_i}F(y;\hat\mu_i)$ and $b_i=F(y;\hat\mu_i)$(see \citep{Sadeghpour16}).\\
When the fitted model is correct, the rqrs are normally distributed which can be tested using the Shapiro-Wilk test.

\section{Three Inflated Poisson Copoun Distribution} \label{sec:ThI}
The problem of zero moderation in count data has a long history in statistical literature. \citep{Neyman39} originally proposed the idea of zero-inflation when there are excess number of zeros present in the data. The ZIP distribution was first proposed by \citep{Mullahy86} as a combination of Poisson distribution and a distribution to a point mass at zero, with mixing probability $\alpha$. The pmf of ZIP is given by
\begin{equation}
p(y;\lambda,\alpha)=\begin{cases}
\alpha+(1-\alpha)e^{-\lambda}; \; y=0\\
(1-\alpha)\frac{e^{-\lambda}\lambda^y}{y!};\; y>0
\end{cases}
\end{equation}
where $\alpha$ is the zero-inflation parameter $(0<\alpha<1)$, $\lambda>0$ and if $\alpha$=0, ZIP reduces to Poisson distribution.\\
Let's suppose that Y be a random variable following Three Inflated Poisson Copoun distribution, denoted by ThIPCD($\eta,\phi,\alpha$) if its probability mass function is given by
\begin{equation}
P(y;\eta,\phi,\alpha)=P(Y=y)=\begin{cases}
\alpha+(1-\alpha)\frac{\eta^2}{(\phi+\eta)(1+\eta)^7}\left[(1+\eta)^3+20\phi\eta^2\right]  \qquad ;y=3 \\
(1-\alpha)\frac{\eta^2}{(\phi+\eta)(1+\eta)^{y+4}}\left[(1+\eta)^3+\frac{\phi\eta^2}{6}(y+1)(y+2)(y+3)\right]\;;y\neq3
\end{cases} \label{eq:pmfThIPCD}
\end{equation}
where $\alpha$ is three inflation parameter $(0<\alpha<1)$ and $(\eta,\phi>0)$.\\
\begin{rem}
(1) when $\alpha$ =0, ThIPCD reduces to PCD.\\
(2) when $\phi$ = 0, ThIPCD reduces to geometric distribution with inflated mass at y=3.
\end{rem}
\begin{theorem}
ThIPCD has a higher probability of 3 than a general PCD.
\end{theorem}
\begin{proof}
\begin{eqnarray}
\alpha+(1-\alpha)\frac{\eta^2}{(\phi+\eta)(1+\eta)^7}\left[(1+\eta)^3+20\phi\eta^2\right] >0\ \nonumber \\
\text{and} \quad  \frac{\eta^2}{(\phi+\eta)(1+\eta)^7}\left[(1+\eta)^3+20\phi\eta^2\right] >0 \nonumber\\
\text{Hence}\quad  \alpha\left[1-\frac{\eta^2}{(\phi+\eta)(1+\eta)^7}\left[(1+\eta)^3+20\phi\eta^2\right] \right] >0  \nonumber\\
\Rightarrow \frac{\eta^2}{(\phi+\eta)(1+\eta)^7}\left[(1+\eta)^3+20\phi\eta^2\right] <1 \nonumber
\end{eqnarray}
Multiplying $\alpha$ to b/s we get
\begin{equation}
\alpha \left[\frac{\eta^2}{(\phi+\eta)(1+\eta)^7}\left[(1+\eta)^3+20\phi\eta^2\right]\right] < \alpha \nonumber
\end{equation}
Adding $\alpha$ to b/s we get \\
\begin{equation}
\alpha-\alpha\left[\frac{\eta^2}{(\phi+\eta)(1+\eta)^7}\left[(1+\eta)^3+20\phi\eta^2\right]\right] >0 \nonumber 
\end{equation}\\

Finally adding $\frac{\eta^2}{(\phi+\eta)(1+\eta)^7}\left[(1+\eta)^3+20\phi\eta^2\right]$ to b/s we get

\begin{align}
&\alpha+(1-\alpha)\frac{\eta^2}{(\phi+\eta)(1+\eta)^7}\left[(1+\eta)^3+20\phi\eta^2\right]  >\nonumber \\
& \hspace{2cm}\frac{\eta^2}{(\phi+\eta)(1+\eta)^7}\left[(1+\eta)^3+20\phi\eta^2\right] 
\end{align}
Hence proved.
\end{proof}

\begin{figure}[H]
\centering
\includegraphics[width=12cm,height=12cm]{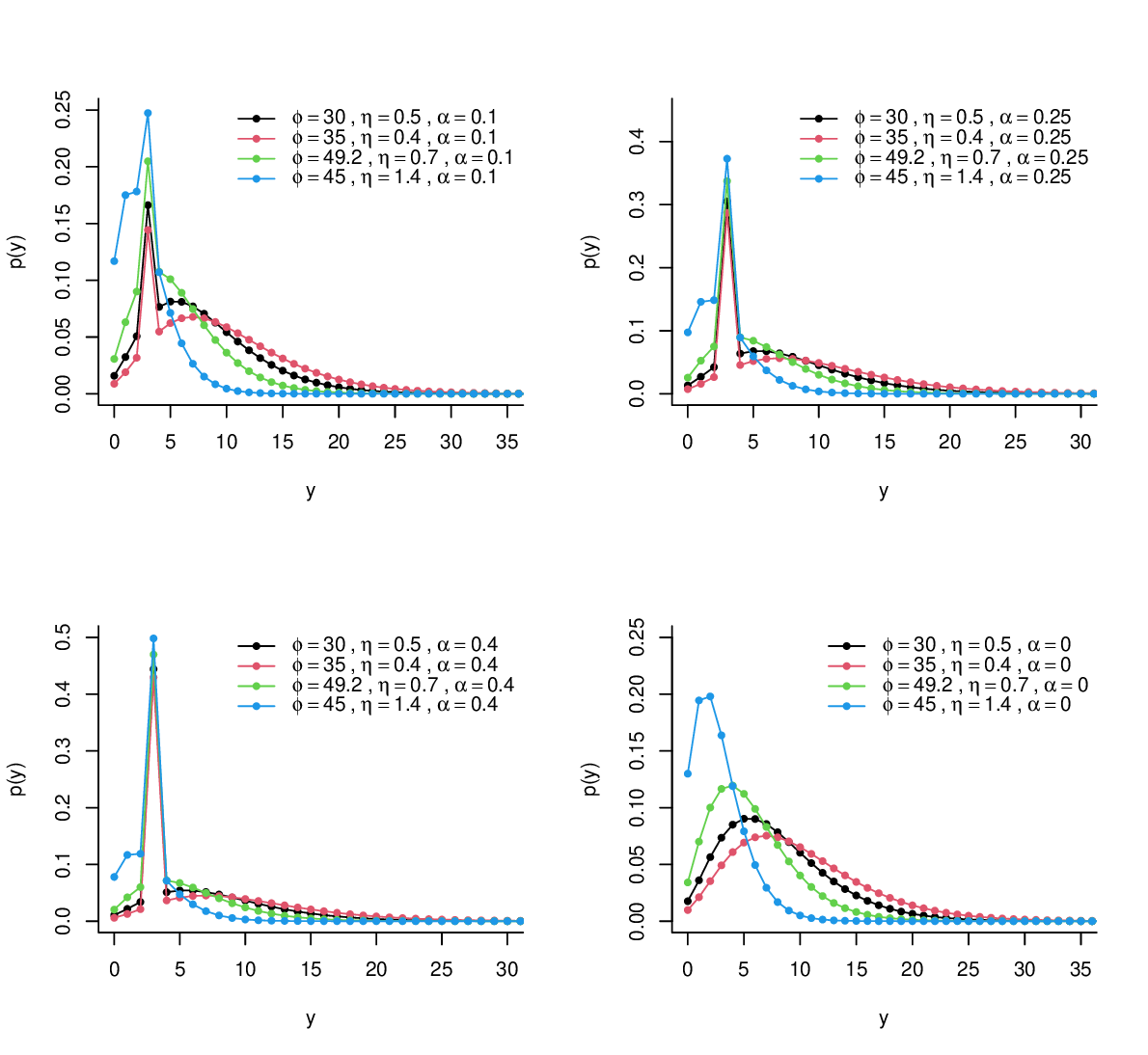}
\caption{The pmf plots of ThIPCD} \label{fig:pmfThIPCD}
\end{figure}

Figure (\ref{fig:pmfThIPCD}) provides the pmf plots of ThIPCD($\eta,\phi,\alpha$) with various parameter values of $\eta$, $\phi$, and $\alpha$. Figure \ref{fig:pmfThIPCD} shows that as $\alpha$ increases, the curve peaks at y=3, and then tends to normal as $\phi$ increases, $\eta<1$, and $\alpha$ decreases.

\section{Distributional Properties} \label{sec:dp}
\subsection{Moments}
If Y $\sim$ ThIPCD($\eta,\phi,\alpha$), then the $r^{th} $order raw moment is obtained as
\begin{equation}
\mu_r'=E(y^r)=\sum_{y=0}^{\infty}y^rP(y;\eta,\phi,\alpha) \nonumber
\end{equation}
\begin{align}
&=3^r\left[\alpha+(1-\alpha)\Bigg(\frac{\eta^2}{(\phi+\eta)(1+\eta)^7}\left[(1+\eta)^3+20\phi\eta^2\right]\Bigg)\right] \nonumber\\ 
 & \hspace{1.4cm}+ \sum_{y\neq 3}y^r(1-\alpha)\frac{\eta^2}{(\phi+\eta)(1+\eta)^{y+4}}\left[(1+\eta)^3+\frac{\phi\eta^2}{6}(y+1)(y+2)(y+3)\right] \nonumber
\end{align}
\begin{align}
&=3^r\alpha+ (1-\alpha)\sum_{y=0}^{\infty}y^r\frac{\eta^2}{(\phi+\eta)(1+\eta)^{y+4}}\left[(1+\eta)^3+\frac{\phi\eta^2}{6}(y+1)(y+2)(y+3)\right] \nonumber
\end{align}
In particular, 
\begin{eqnarray}
\mu_1'&=&3\alpha+(1-\alpha)\left[\frac{\eta+4\phi}{\eta(\phi+\eta}\right] \nonumber\\
\mu_2'&=&9\alpha+(1-\alpha)\left[\frac{\eta^2+2(2\phi+1)\eta+20\phi}{\eta^2(\phi+\eta)}\right] \nonumber\\
\mu_3'&=& 27\alpha+(1-\alpha)\left[\frac{\eta^3+2(2\phi+3)\eta^2+6(10\phi+1)\eta+120\phi}{\eta^3(\phi+\eta)}\right] \nonumber\\
\mu_4'&=& 81\alpha+(1-\alpha)\left[\frac{\eta^4+2(2\phi+7)\eta^3+4(35\phi+9)\eta^2+24(30\phi+1)\eta+840\phi}{\eta^4(\phi+\eta)}\right] \nonumber
\end{eqnarray}

\begin{align}
\therefore V(X)=9\alpha+(1-\alpha)+\frac{(1-\alpha)}{\eta^2(\phi+\eta)^2}\Bigg[\begin{aligned}&(1-6\alpha)\eta^2+(5\phi-30\alpha\phi+\alpha+1)\eta^2\\&+2(2\phi^2+7\phi+4\phi\alpha-12\phi^2\alpha)\eta
+4(1+4\alpha)\phi^2\end{aligned}\Bigg] \nonumber
\end{align}
\subsection{Probability Generating function of ThIPCD}
\begin{theorem}
If Y $\sim$ ThIPCD($\eta, \phi, \alpha$), then its PGF $P_Y(S)$ is as follows
\begin{align}
P_Y(S)=\alpha s^3+(1-\alpha)\frac{\eta^2}{\phi+\eta}\left[\frac{(\eta-s+1)^3+\phi\eta^2}{(\eta-s+1)^4}\right]  \label{eq:pgfThIPCD}
\end{align}
\end{theorem}

\begin{proof}
If Y$\sim$ ThIPCD($\eta, \phi, \alpha$), then its PGF $P_Y(S)$ is

\begin{eqnarray}
P_Y(S)&=&E(S^y)=\sum_{y=0}^{\infty}P(y;\eta,\phi,\alpha)S^y \nonumber \\
&=& \alpha s^3 + (1-\alpha)\frac{\eta^2}{\phi+\eta} \sum_{y=0}^{\infty}{\frac{s^y\left[(1+\eta)^3+\frac{\phi\eta^2}{6}(y_i+1)(y_i+2)(y_i+3)\right]}{(1+\eta)^{y+4}}} \nonumber \\
&=& \alpha s^3+(1-\alpha)\frac{\eta^2}{\phi+\eta}\Bigg[\frac{(\eta-s+1)^3+\phi\eta^2}{(\eta-s+1)^4}\Bigg] 
\end{eqnarray}
Hence proved.
\end{proof}

\begin{rem}
Putting $S=e^t$ in \ref{eq:pgfThIPCD}, the Moment Generating Function $M_Y(t)$ of Y $\sim$ ThIPCD($\eta,\phi,\alpha$) is as follows
\begin{align}
M_Y(t)=\alpha e^{3t}+(1-\alpha)\frac{\eta^2}{\phi+\eta}\Bigg[\frac{(\eta-e^t+1)^3+\phi\eta^2}{(\eta-e^t+1)^4}\Bigg]
\end{align}
\end{rem}
\begin{rem}
 Putting $S=e^{it}$ in \ref{eq:pgfThIPCD}, the Characteristic Function $\varphi_Y(t)$ of Y $\sim$ ThIPCD($\eta,\phi,\alpha$) is as follows
\begin{align}
\varphi_Y(t)=\alpha e^{3it}+(1-\alpha)\frac{\eta^2}{\phi+\eta}\Bigg[\frac{(\eta-e^{it}+1)^3+\phi\eta^2}{(\eta-e^{it}+1)^4}\Bigg]
\end{align}
\end{rem}

\section{Maximum Likelihood Estimation} \label{sec:mle}
Using the MLE approach, the parameters $\alpha$, $\eta$, and $\phi$ of equation (\ref{eq:pmfThIPCD}) can be derived as follows:\\
Let $y_1, y_2, y_3,..., y_n$ be a random sample from ThIPCD and let for $i=1,2,3,...,n$

\begin{align}
a_i=
\begin{cases}
1, if y_i=3\\
0, otherwise
\end{cases}
\end{align}
Then equation(\ref{eq:pmfThIPCD}) can be described as follows

\begin{align}
P(Y=y_i)&=\Bigg[\alpha+(1-\alpha)\frac{\eta^2}{(\phi+\eta)(1+\eta)^7}\left[(1+\eta)^3+20\phi\eta^2\right]\Bigg]^{a_i} \nonumber\\
& \hspace{2cm} \Bigg[(1-\alpha)\frac{\eta^2}{(\phi+\eta)(1+\eta)^{y_i+4}}\left[(1+\eta)^3+\frac{\phi\eta^2}{6}(y_i+1)(y_i+2)(y_i+3)\right]\Bigg]^{1-a_i}
\end{align}
Hence the likelihood function, $L=L(\eta, \phi, \alpha; y_1, y_2, y_3,..., y_n)$ will be
\begin{equation}
L=
\prod_{i=1}^{n}\left\{\!\begin{aligned}
&\Bigg[\alpha+(1-\alpha)\frac{\eta^2}{(\phi+\eta)(1+\eta)^7}\left[(1+\eta)^3+20\phi\eta^2\right]\Bigg]^{a_i} \nonumber \\[1ex]
& \Bigg[(1-\alpha)\frac{\eta^2}{(\phi+\eta)(1+\eta)^{y_i+4}}\left[(1+\eta)^3+\frac{\phi\eta^2}{6}(y_i+1)(y_i+2)(y_i+3)\right]\Bigg]^{1-a_i}
\end{aligned}\right\} \nonumber
\end{equation}

\begin{align}
\Rightarrow L&=\Bigg[\alpha+(1-\alpha)\frac{\eta^2}{(\phi+\eta)(1+\eta)^7}\left[(1+\eta)^3+20\phi\eta^2\right]\Bigg]^{n_0} \nonumber  \\
& \hspace{4cm} \prod_{i=1}^{n}\Bigg[(1-\alpha)\frac{\eta^2}{(\phi+\eta)(1+\eta)^{y_i+4}}\left[(1+\eta)^3+\frac{\phi\eta^2}{6}(y_i+1)(y_i+2)(y_i+3)\right]\Bigg]^{k_i} \nonumber
\end{align}
where $k_i=1-a_i$, $n_0=\sum_{i=1}^{n}a_i$, and $n_0$ represents the number of 3's in the sample.\\

Therefore 
\begin{align}
\log L&=n_0\log\Bigg[\alpha+(1-\alpha)\frac{\eta^2}{(\phi+\eta)(1+\eta)^7}\left[(1+\eta)^3+20\phi\eta^2\right]\Bigg]+
(n-n_0)\log(1-\alpha)+2(n-n_0)\log\eta \nonumber \\ &
-(n-n_0)\log(\phi+\eta)+\sum_{i=1}^{n}k_i\log\Bigg[(1+\eta)^3+\frac{\phi\eta^2}{6}(y_i+1)(y_i+2)(y_i+3)\Bigg]-\sum_{i=1}^{n}k_i(y_i+4)\log(1+\eta)
\end{align}
\begin{equation}
\Rightarrow \frac{\partial\log L}{\partial\alpha}=\frac{n_0\Bigg((\phi+\eta)(1+\eta)^7-\eta^2\Big[(1+\eta)^3+20\phi\eta^2\Big]\Bigg)}
{\alpha(\phi+\eta)(1+\eta)^7+(1-\alpha)\eta^2\Big[(1+\eta)^3+20\phi\eta^2\Big]}
\end{equation}

\section{Empirical studies} \label{sec:app}

\subsection{Application of PCD in health sciences}
This data set is related to the adverse events to the Anthrax Vaccine Absorbed (AVA) vaccine that were recorded during clinical trials. This data set has been taken from \citep{Rose06}. The data states that four study injections are given to 1005 research participants at intervals of 0 weeks, 2 weeks, 4 weeks, and 6 months, and a total of 4020 observations are collected. The results are given in Table (\ref{tab:vaccgof}).

\begin{table}[H]
	\centering
	%\fontsize{6.5}{11}\selectfont
	\caption{Goodness of fit} \label{tab:vaccgof}
	\begin{tabular}{lrrrrrr} 
		\hline	
		Model&$-l$&$AIC$&$BIC$&$Chi-sq$&df&$p$-value\\ \hline
		PCD & 6737.216 & 13478.432 &13491.030 & 4.52 & 7 & 0.72 \\
		PD & 7231.135 &14464.269 &14470.568 & 1279.25 & 5 & \textless0.001\\
		ZIP & 6868.794 & 13741.588 & 13754.186 & 1101.35 & 5 & \textless0.001\\
		GD & 6778.044 &13558.089 &13564.388 & 74.28 & 10 &\textless0.001\\
		NB & 6740.605 &13485.210 &13497.808 &10.04 & 8 & 0.255\\
		PL & 6745.994 & 13493.987 & 13500.286 & 19.24 & 10 & 0.037\\
		\hline
	\end{tabular}  
\end{table}  
The results in Table (\ref{tab:vaccgof}) are calculated in R software using  $"fitdistrplus"$ package. We have compared our proposed model i.e., PCD to other discrete distributions including Poisson distribution (PD) \citep{Poisson37}, Zero-Inflated Poisson distribution (ZIP) \citep{Lambert92}, Geometric distribution (GD), Negative Binomial distribution (NB) and Poisson Lindley distribution (PL) \citep{Sankaran70}. In order to compare the models, we have taken into account different criterions like negative log-likelihood $(-l)$, Akaike Information Criteria $(AIC)$, and Bayesian Information Criteria $(BIC)$. When compared to alternative models, the results in Table (\ref{tab:vaccgof}) demonstrate that the PCD model offers a very good fit to the data. In addition, Table (\ref{tab:vaccgof}) shows that, in comparison to other discrete distributions, our proposed model has the lowest AIC and BIC values in addition to a strong $p$-value. Therefore, it can be said that the PCD model is superior for modelling data that is over-dispersed, especially in health sciences.

\subsection{Application of PCD Linear Model}
The data set that we consider here is related to the number of infected blood cells on microscope slides (per $mm^2$) collected from 511 randomly chosen participants (see \citep{Crawley12}). The data was recently used by \citep{Hassan22} and \citep{Wani23}. This study aims to examine the impact of variables, including an individual's age, weight, smoking status, and sex, on the response variable "number of infected blood cells(Y)" reported.\\
The fitted regression model is
\begin{align}
\mu_i=\exp(\beta_0+\beta_1 x_{1i}+\beta_2 x_{2i}+\beta_3 x_{3_i}+\beta_4 x_{4i}+\beta_5 x_{5i}+\beta_6 x_{6_i})
\end{align}

Figure (\ref{fig:histpcd}) displays the histogram of the response variable. Since, the value of  $DI$ of the response variable is calculated as 2.39. As a result, the response variable is over-dispersed. Hence, over-dispersion should be manageable for the model that explains the current data set.

The accurate model for the data is found using the model selection criteria, which include $(-l)$, $AIC$, and $BIC$ values. The best model is shown by the statistics having smallest values. Table (\ref{tab:regtab}) presents the findings from the Poisson distribution, Negative Binomial distribution, Two Parameter Poisson XGamma distribution $(TPPXG)$ \citep{Wani23} and PCD regression models. The PCD regression model is better fitted than the other three count regression models since it has the lowest values of the model selection criteria.

The results of the PCD regression allow us to conclude that smoking, being overweight, and being obese; all increase the number of infected blood cells in a patient.

\begin{table}[H]
	\centering 
	\setlength{\tabcolsep}{1pt}
	\fontsize{10}{10}\selectfont
	\renewcommand{\arraystretch}{2} 
	\caption{The results of the regression models for the number of infected blood cells.}\label{tab:regtab}
	\begin{tabular}{|c|cccc|cccc|cccc|cccc|}
	\hline
	Covariates&\multicolumn{4}{|c|}{$Poisson$}&\multicolumn{4}{|c|}{$Negative Binomial$}&\multicolumn{4}{|c|}{$TPPXG$}&\multicolumn{4}{|c|}{$PCD$}  \\ \cline{2-17}
			&$ Est$ & $SE$ &$z$-$value$& $p$-$value$ & $Est$ & $SE$ &$z$-$value$& $p$-$value$ &$ Est$ & $SE$ &$z$-$value$ & $p$-$value $&$ Est$ & $SE$ &$z$-$value$ & $p$-$value $ \\ 
			\hline
			Intercept & -1.17 & 0.12 & -9.50 & <0.001  & -1.13 & 0.16 & -6.90 & <0.001  & -1.15 & 0.16 & -7.31 & <0.001 & -1.15 & 0.16 & -7.27 & <0.001 \\ 
			Sex & 0.12 & 0.11 & 1.13 & 0.26 & 0.20 & 0.16 & 1.22 & 0.22  & 0.18 & 0.15 & 1.20 & 0.23 &0.17 & 0.15 & 1.17 & 0.24\\ 
			Smoking & 1.31 & 0.11 & 11.70 & <0.001 & 1.18 & 0.17 & 7.08 & <0.001  & 1.22 & 0.15 & 8.21 & <0.001 &1.22 & 0.15 & 8.20 & <0.001 \\ 
			Young & -0.06 & 0.12 & -0.48 & 0.63 & 0.06 & 0.20 & 0.30 & 0.77  & 0.03 & 0.17 & 0.17 & 0.87 &0.03 & 0.17 & 0.17 & 0.86 \\ 
			Mid-Age& -0.06 & 0.13 & -0.44 & 0.66 & -0.04 & 0.19 & -0.22 & 0.83 & -0.04 & 0.17 & -0.25 & 0.80  &-0.04 & 0.17 & -0.25 & 0.79 \\ 
			Overweight & 0.53 & 0.13 & 4.08 & <0.001  & 0.49 & 0.18 & 2.67 & 0.01 &   0.50 & 0.17 & 2.89 & <0.001 &0.49 & 0.17 & 2.87 & <0.001\\ 
			Obese & 0.93 & 0.12 & 7.86 & <0.001  & 0.80 & 0.17 & 4.66 & <0.001& 0.84 & 0.16 & 5.27 & <0.001  & 0.85 & 0.16 & 5.32 & <0.001\\ 
			$-l$ &\multicolumn{4}{|c|}{658.70}&\multicolumn{4}{|c|}{617.06}  &\multicolumn{4}{|c|}{612.18} &\multicolumn{4}{|c|}{609.44}\\ 
			$AIC$ & \multicolumn{4}{|c|}{1331.39} &\multicolumn{4}{|c|}{1250.12}&\multicolumn{4}{|c|}{1240.37} &\multicolumn{4}{|c|}{1234.88}  \\ 
			$BIC$ &\multicolumn{4}{|c|}{1361.05}& \multicolumn{4}{|c|}{1284.01} &\multicolumn{4}{|c|}{1274.26}& \multicolumn{4}{|c|}{1268.78} \\
			\hline
	\end{tabular}
\end{table}

Figure (\ref{fig:qqplot}) shows the rqrs and the Normal quantile–quantile (QQ) plots of the PCD regression models. Figure (\ref{fig:profilelog}) shows the profile log-likelihood functions of the parameters of PCD regression model. These figures demonstrate that the estimated parameters are actual maximisers of $l$.

\begin{figure}[H]
\centering
\includegraphics[width=15cm,height=10cm]{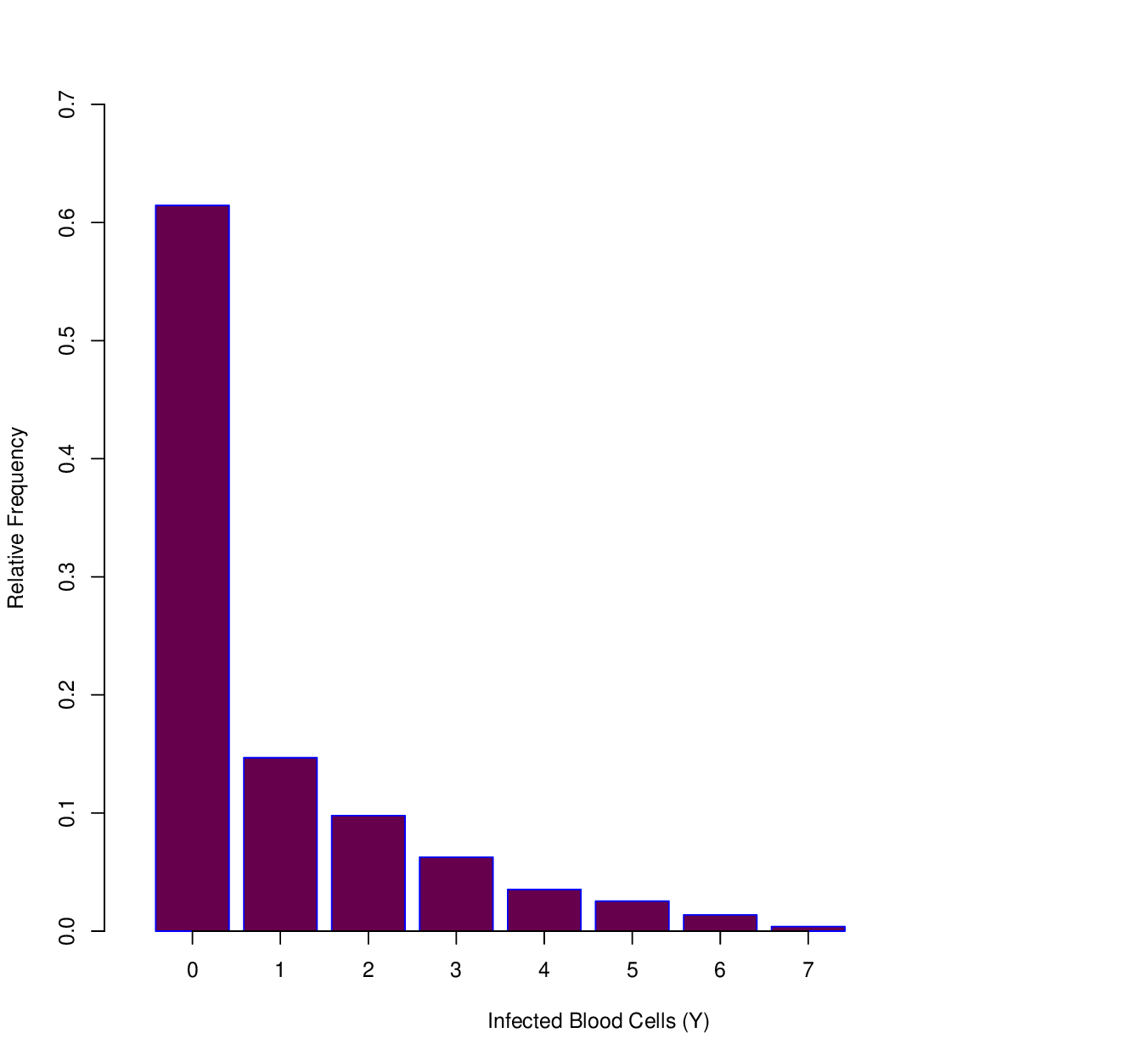}
\caption{The barplot of the response variable} \label{fig:histpcd}
\end{figure}

\begin{figure}[H]
\centering
\includegraphics[width=12cm]{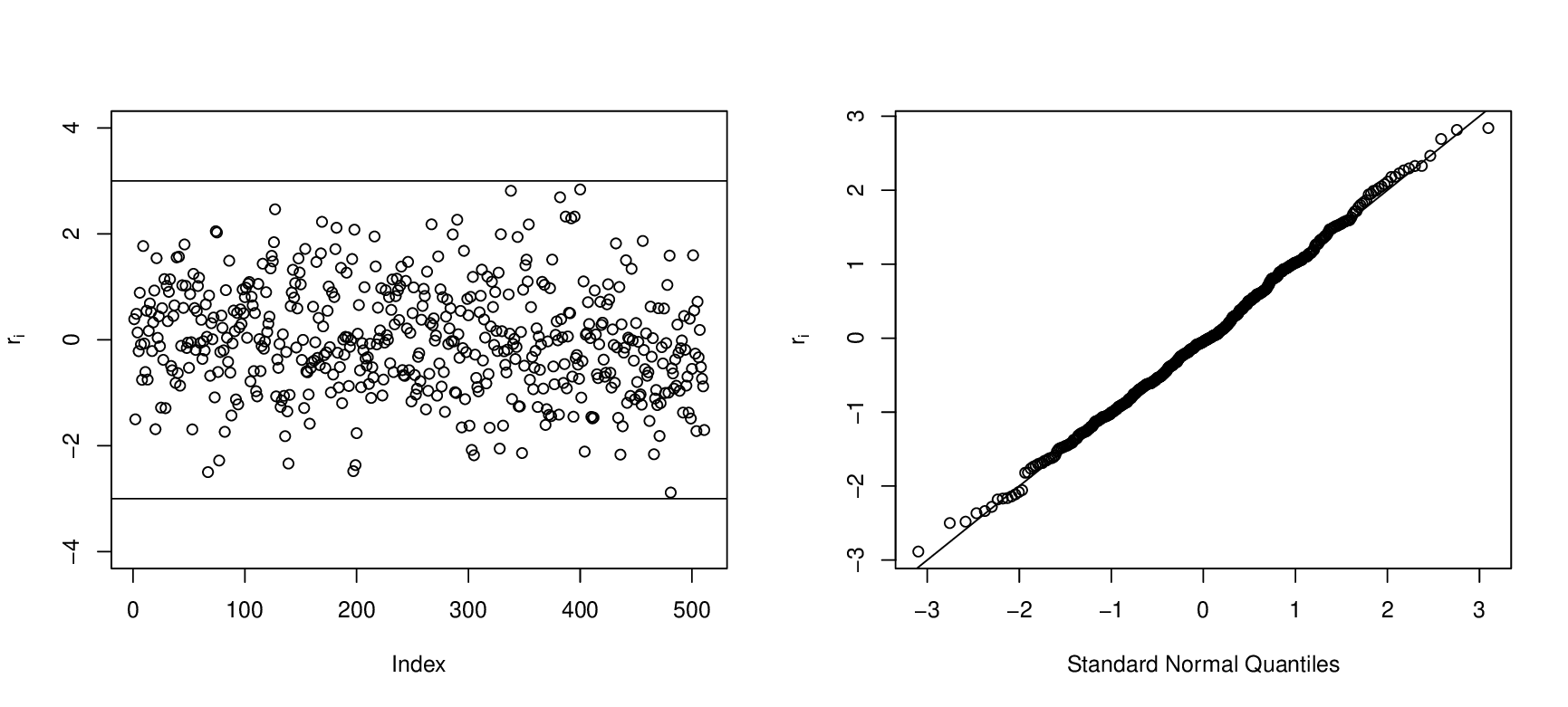}
\caption{The rqrs and the Normal Q-Q plot.} \label{fig:qqplot}
\end{figure}

\begin{figure}[H]
\centering
\includegraphics[width=15cm,height=15cm]{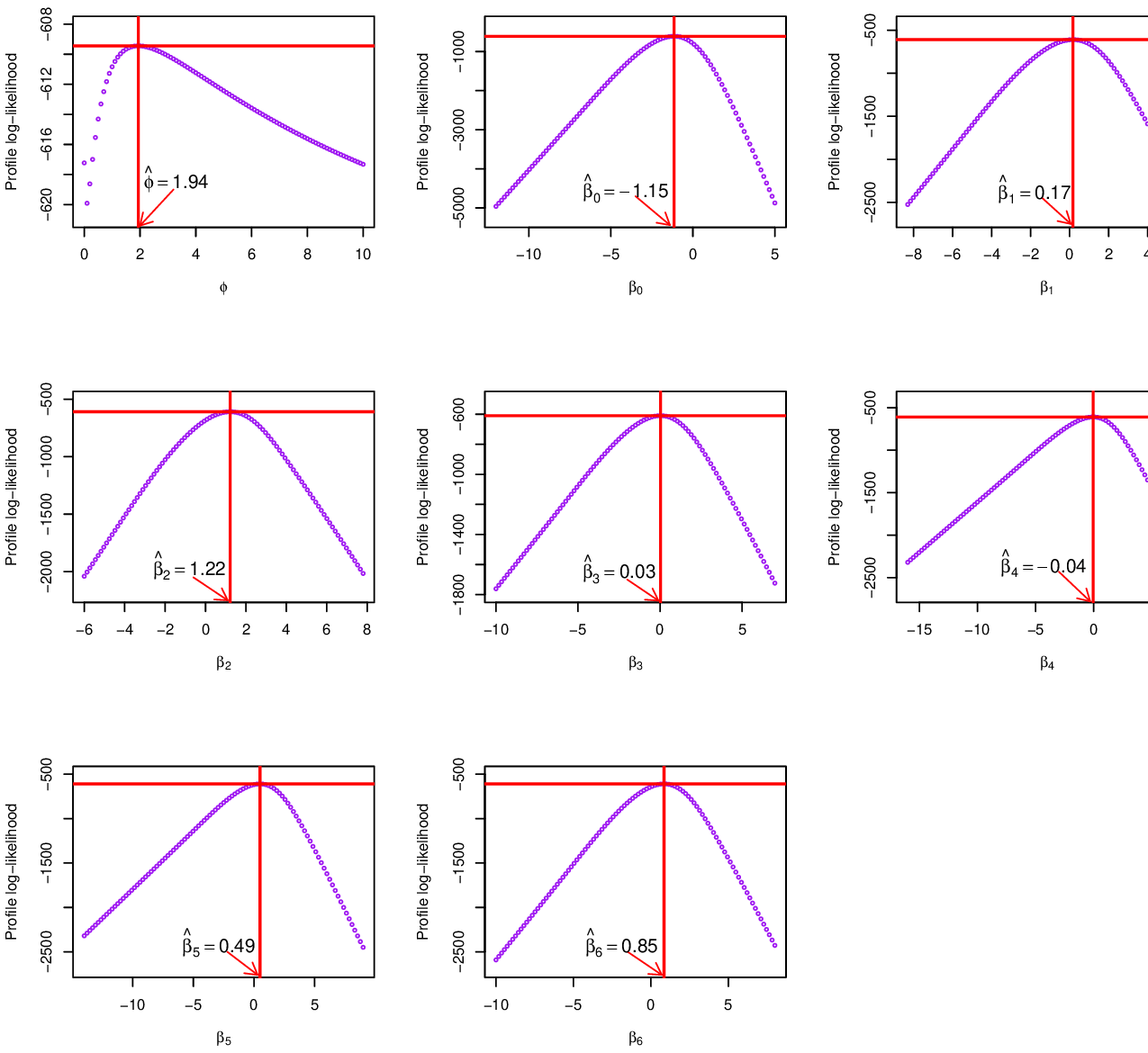}
\caption{The profile log-likelihood plots of the PCD regression model for the infected cells data set.} \label{fig:profilelog}
\end{figure}

\subsection{Application of ThIPCD in hospital LOS}
One important measure of hospital care management effectiveness, care costs, and hospital planning is hospital LOS. Hospital LOS is frequently used as a gauge for the results of medical procedures, as a reference for the advantages of a particular treatment, or as a significant risk factor for unfavourable outcomes. Thus, it is always crucial for healthcare professionals to understand hospital LOS variability. Hospital LOS data can be thought of as count data since it has discrete, non-negative values, usually skewed to the right, and frequently has too many zeros.
\citep{Orooji21}  studied the factors affected LOS among elderly patients using count regression models. \citep{Fernandez22} compared different statistical methods for modelling count data with application to hospital LOS.

We have considered the data set reported by \citep{Yau03}. This data set is related to the LOS of 261 patients who have pancreas disorder in Werstern Australia (1998-1999). We have compared ThIPCD with Three Inflated Poisson distribution (ThIPD) given by \citep{Rahman22}. The data and the results are given in Table (\ref{tab:appThIPCD}).

\begin{table}[H]
\centering
	\fontsize{10}{5}\selectfont
	\setlength{\tabcolsep}{4.5pt}
	\renewcommand{\arraystretch}{1.8}
	\caption{Results of LOS on ThIPCD \citep{Yau03}} \label{tab:appThIPCD}
	\begin{tabular}{ccccc}
  \hline
  LOS & Frequency & ThIPD & ThIPCD \\ 
  \hline
0 & 45 & 11.26 & 42.43  \\ 
  1 & 35 & 35.39 & 39.91 \\ 
 2 & 35 & 55.60 & 38.21 \\ 
  3 & 47 & 58.49 & 47.00  \\ 
  4& 40 & 45.74 & 27.67 \\ 
  5 & 20 & 28.74 & 21.10  \\ 
  6 & 13 & 15.05 & 15.24 \\ 
  7 & 8 & 6.75 & 10.54  \\ 
  8 & 4 & 2.65 & 7.03 \\ 
  9 & 5 & 0.93 & 4.56  \\ 
  10 & 3 & 0.29 & 2.88  \\ 
  11 & 1 & 0.08 & 1.79  \\ 
  12 & 4 & 0.02 & 1.09  \\ 
  13 & 0 & 0.01 & 0.65  \\ 
  14 & 1 & 0.00 & 0.39  \\ 
  \multicolumn{2}{r}{$df$} & 5.00 & 5.00 \\ 
    \multicolumn{2}{r}{$\chi^2$ value}  & 136.38 & 8.81 \\ 
    \multicolumn{2}{r}{$p-value$}  & <0.001 & 0.185\\\\
     
   \multicolumn{2}{r}{-$l$}  & 640.79  & 579.62\\ 
   \multicolumn{2}{r}{$AIC$}  & 1285.58 & 1165.25\\ 
    \multicolumn{2}{r}{$BIC$} & 1292.70  & 1175.94 \\ 
   \hline
\end{tabular}
\end{table}
From Table (\ref{tab:appThIPCD}), it is evident that ThIPCD model provides a better fit than ThIPD as can be seen from the values of $AIC$ and $BIC$. $p$-value also suggests the same.

\section{Conclusion} \label{sec:con}
Finding new probability models to handle over-dispersed data is crucial in addressing challenges posed by over-dispersion. Appropriately selecting such models for various scenarios leads to better data fits and more reliable interpretations. In this context, we introduced a novel count data model, the Poisson-Copoun distribution, along with its regression framework and a three inflated version. The model combines the Poisson and Copoun distributions, offering a flexible alternative for count data analysis. Estimation methods, including maximum likelihood estimation and the method of moments, have been employed, and explicit formulations for key statistical measures of the proposed model have been derived.

The utility of the proposed distribution was demonstrated by fitting it to a dataset on vaccine adverse events, where it outperformed competing models based on different model comparison criteria. A regression model based on the PCD was also developed, and its applicability was illustrated through the analysis of infected blood cell counts ($\text{per mm}^2$). Both AIC and BIC criteria indicated that the proposed regression model outperformed the Poisson, Negative Binomial, and TPPXG regression models. Furthermore, residual analysis using randomized quantile residuals validated the robustness of the fitted model.

Additionally, a three-inflated version of the PCD was developed and applied to analyze the length of stay (LOS) data of patients with pancreatic disorders. The results highlighted the effectiveness of the three-inflated models in hospital LOS data analysis.

Given its desirable properties and superior performance, the proposed model is expected to gain traction in analyzing count data across diverse domains, particularly in health sciences. The results suggest that these models could provide better insights into health science data, and practitioners are advised to consider using these models to achieve improved outcomes in their analyses. 

Future research could explore its application to other fields and extend its functionality by incorporating covariates or developing multivariate versions. For instance, an immediate avenue for future work is to study the LOS of tourists and the factors influencing it using three-inflated count distributions.

%\bibliographystyle{apalike}
%\bibliography{bibliography}

\end{document}